\documentclass[letterpaper,10pt,conference]{ieeeconf}

\IEEEoverridecommandlockouts
\usepackage{graphicx}

\usepackage{dblfloatfix}

\usepackage{placeins}

\usepackage{amsmath}
\usepackage{amssymb}
\usepackage{amsfonts}

\usepackage{amsthm}

\usepackage{mathtools}
\usepackage{bm}

\usepackage{booktabs}
\usepackage{makecell}

\usepackage{algorithm}

\usepackage{algpseudocode}

\usepackage{listings}

\usepackage{xcolor}

\usepackage{etoolbox}

\usepackage{cite}

\usepackage{todonotes}

\usepackage[hidelinks]{hyperref}

\newcommand{\MethodFigureWidth}{\linewidth}

\allowdisplaybreaks
\makeatletter
\renewcommand{\@makecaption}[2]{%
  \ifx\@captype\@IEEEtablestring
    \begin{center}{\small #1}\\{\small\scshape #2}\end{center}%
    \@IEEEtablecaptionsepspace
  \else
    \@IEEEfigurecaptionsepspace
    \setbox\@tempboxa\hbox{\small #1.~~ #2}%
    \ifdim\wd\@tempboxa>\hsize
      \setbox\@tempboxa\hbox{\small #1.~~ }%
      \parbox[t]{\hsize}{\small\noindent\unhbox\@tempboxa#2}%
    \else
      \hbox to\hsize{\small\box\@tempboxa\hfil}%
    \fi
  \fi}
\makeatother
\newtheorem{proposition}{Proposition}

\theoremstyle{definition}

\newtheorem{assumption}{Assumption}

\theoremstyle{remark}
\newtheorem{remark}{Remark}

\BeforeBeginEnvironment{proposition}{\vspace{5pt}}
\AfterEndEnvironment{proof}{\vspace{-4pt}}

\makeatletter
\g@addto@macro\@floatboxreset{\normalcolor}
\makeatother

\definecolor{codegreen}{rgb}{0,0.6,0}
\definecolor{codegray}{rgb}{0.5,0.5,0.5}
\definecolor{codepurple}{rgb}{0.58,0,0.82}
\definecolor{backcolour}{rgb}{0.95,0.95,0.92}

\lstdefinestyle{mystyle}{
    backgroundcolor=\color{backcolour},
    commentstyle=\color{codegreen},
    keywordstyle=\color{magenta},
    numberstyle=\tiny\color{codegray},
    stringstyle=\color{codepurple},
    basicstyle=\ttfamily\footnotesize,
    breakatwhitespace=false,
    breaklines=true,
    captionpos=b,
    keepspaces=true,
    numbers=left,
    numbersep=5pt,
    showspaces=false,
    showstringspaces=false,
    showtabs=false,
    tabsize=2
}

\title{\LARGE \bf
SD-DPC: Sparse Dictionary Differentiable Predictive Control
}

\author{
Ali Reza Daneshvar Garmroodi$^{1}$ and Jan Drgo\v{n}a$^{2}$%
\thanks{*This research is supported by the U.S. DOE, Office of Science, ASCR program under the Scientific Discovery through Advanced Computing (SciDAC) Institute “LEADS: LEarning-Accelerated Domain Science”. }%
\thanks{$^{1}$Concordia University, Montreal, QC H3G1M8, Canada.
{\tt\small alireza.daneshvargarmroodi@mail.concordia.ca}}%
\thanks{$^{2}$Johns Hopkins University, Baltimore, Maryland 21218, USA.
{\tt\small jdrgona1@jh.edu}}%
}

\begin{document}
\bstctlcite{IEEEexample:BSTcontrol}
\maketitle
\thispagestyle{empty}
\pagestyle{empty}

%%%%%%%%%%%%%%%%%%%%%%%%%%%%%%%%%%%%%%%%%%%%%%%%%%%%%%%%%%%%%%%%%%%%%%%%%%%%%%%%
% ABSTRACT
%%%%%%%%%%%%%%%%%%%%%%%%%%%%%%%%%%%%%%%%%%%%%%%%%%%%%%%%%%%%%%%%%%%%%%%%%%%%%%%%

\begin{abstract}
We present sparse dictionary differentiable predictive control (SD-DPC), a
framework for learning sparse, interpretable feedback policies for nonlinear
systems from data. A prediction model is first identified by rollout-based
sparse identification of nonlinear dynamics (SINDy), building on gradient-based
and multistep formulations. The policy is then parameterized as a sparse
combination of dictionary functions and trained by differentiating a
constrained finite-horizon predictive-control objective through this model, so
that its terms are selected by closed-loop performance rather than by
imitating a previously trained controller. The result is an explicit feedback
law with only a handful of terms. Across three benchmark control problems, SD-DPC satisfies the constraints in {almost} all test scenarios, {outperforms a policy distilled onto the same terms by up to an order of magnitude, and requires orders of magnitude less memory and
online computation than an optimization benchmark,} while {admitting explicit
sensitivity bounds}.
\end{abstract}

%%%%%%%%%%%%%%%%%%%%%%%%%%%%%%%%%%%%%%%%%%%%%%%%%%%%%%%%%%%%%%%%%%%%%%%%%%%%%%%%
% INTRODUCTION
%%%%%%%%%%%%%%%%%%%%%%%%%%%%%%%%%%%%%%%%%%%%%%%%%%%%%%%%%%%%%%%%%%%%%%%%%%%%%%%%

\section{INTRODUCTION}

Model predictive control (MPC) optimizes closed-loop performance under state
and input constraints, but requires solving a constrained optimization problem
at every sampling instant \cite{GarciaPrett1989,Mayne2000,Schwenzer2021},
which can be prohibitive for fast nonlinear systems or restricted hardware.
Explicit MPC moves the optimization offline, although the number of regions
describing the solution can grow rapidly with the problem dimension
\cite{Bemporad2002}. Learned policies scale better, either by approximating a
precomputed MPC law~\cite{Hertneck2018,Karg2020,LiHuaCao2022} or by directly minimizing a predictive-control objective
through differentiable closed-loop rollouts
\cite{Drgona2022,Drgona2024}. The latter
approach, differentiable predictive control (DPC), yields an explicit policy
that is inexpensive to evaluate online. However, these learned policies are commonly
deep neural networks, which distribute the control law across many coupled
parameters and are difficult to interpret, verify, modify, and adapt after deployment
\cite{Hose2024}.

\begin{figure}[!t]
    \centering
    \includegraphics[width=\MethodFigureWidth]{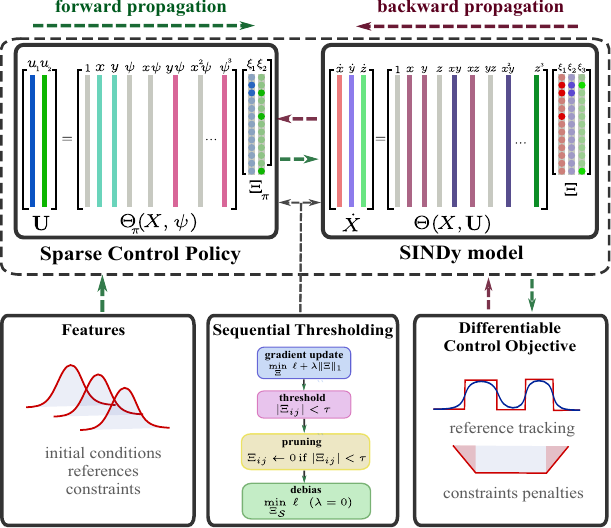}
    \caption{Pipeline: sparse system identification and differentiable
    predictive control policy learning.}
    \label{fig:sindy-dpc}
\end{figure}

Sparse identification of nonlinear dynamics (SINDy) represents system dynamics
as a sparse combination of candidate functions \cite{Brunton2016a}. Extensions
to controlled systems enable compact models for model-based control
\cite{brunton2016sparse, Kaiser2018}, while ensemble, weak-form, gradient-based, and
trajectory-based formulations improve robustness to noise and avoid pointwise
derivative estimates
\cite{Messenger2021,Fasel2022,Champion2019,Kaheman2022,ZhangTangHe2022}. {The same dictionary structure can also represent a feedback law, which is
the idea pursued here.}

Sparse and symbolic representations have also been used for interpretable
control policies. One line of work optimizes the policy structure through
model-free reinforcement learning \cite{Landajuela2021,Guo2023,Botteghi2024},
where constraints enter only through the reward. Another trains a neural policy
first and then fits a sparse policy to its actions \cite{Zolman2025}, so the
result is selected for agreement with the teacher rather than for closed-loop
performance. Sparse polynomial feedback laws have also been learned by
minimizing an averaged closed-loop cost with an $\ell_1$ penalty
\cite{Kunisch2023}, but for known control-affine dynamics with a constant input
matrix and without state or input constraints. To our knowledge, the terms of
an explicit sparse policy have not been selected under a constrained
finite-horizon predictive-control objective evaluated through a model
identified from data.

{In this work, we train a sparse dictionary policy directly through DPC.
The prediction model and the feedback law are both sparse combinations of
candidate functions, learned by the same sequentially thresholded gradient
procedure, so the policy terms are selected by the constrained closed-loop
objective itself rather than by an intermediate imitation step. The
contributions of this paper are as follows.}
{\begin{itemize}\setlength{\itemsep}{1pt}\setlength{\parskip}{0pt}
\item \textbf{Closed-loop sparse policy learning framework:} SD-DPC represents both the prediction
model and the feedback policy as sparse dictionary expansions and, to our
knowledge, is the first to select the policy terms under a constrained
predictive-control objective through a model identified from data
\cite{Champion2019,Kaheman2022,ZhangTangHe2022}.
\item \textbf{Structural guarantees:} we show finite termination of the
thresholded training with a bound on the cost of each prune, and derive a
Lipschitz bound of the policy from its retained coefficients that bounds the
closed-loop deviation caused by pruning or model mismatch.
\item \textbf{Closed-loop validation:} on three benchmarks evaluated on the
true plants, SD-DPC satisfies {almost} all constraints {with orders of magnitude
less memory and evaluation time than online optimization}, whereas distilling
a deep neural policy onto the same terms increases the tracking error {by up to
an order of magnitude} and violates the obstacle constraint.
\end{itemize}}

The remainder of the paper is organized as follows. Section~\ref{sec:background}
reviews SINDy and DPC, Section~\ref{sec:sparsedpc} develops SD-DPC{, its
training algorithm, and its structural properties}, Section~\ref{sec:experiments} details the case studies,
Section~\ref{sec:results} presents the results, and
Section~\ref{sec:conclusions} concludes.

\section{Background}
\label{sec:background}

\subsection{Sparse Identification of Nonlinear Dynamics}
\label{subsec:sindy_prelim}

Consider the nonlinear controlled system $\dot{x}(t)=f(x(t),u(t))$ with state
$x\in\mathbb{R}^{n_x}$ and input $u\in\mathbb{R}^{n_u}$, where $f$ is unknown
but trajectory data are available. Let $X\in\mathbb{R}^{{n_s}\times n_x}$
and $U\in\mathbb{R}^{{n_s}\times n_u}$ collect ${n_s}$ state and input
samples, and let $\dot{X}\in\mathbb{R}^{{n_s}\times n_x}$ contain the
corresponding state derivatives, measured or estimated by numerical
differentiation. SINDy approximates $f$ by a small number of candidate
functions collected in a library
\begin{equation}
\Theta(X,U)=\big[\theta_1(X,U)\ \cdots\ \theta_p(X,U)\big]\in\mathbb{R}^{{n_s}\times p},
\label{eq:sindy_library}
\end{equation}
typically polynomial, trigonometric, and state--input interaction terms, so
that $\dot{X}\approx\Theta(X,U)\Xi$ with $\Xi\in\mathbb{R}^{p\times n_x}$. The
coefficients follow from the sparsity-promoting regression
\begin{equation}
\min_{\Xi}\ \tfrac{1}{2}\big\|\dot{X}-\Theta(X,U)\Xi\big\|_F^2+\lambda\,\mathcal{R}(\Xi),
\label{eq:sindy_optimization}
\end{equation}
where $\mathcal{R}$ promotes sparsity and $\lambda>0$ weights it. In the
standard formulation this is solved by sequentially thresholded least squares:
a least-squares fit, removal of coefficients below a threshold, and refitting
on the retained terms, repeated until the retained coefficient set no longer
changes.

\subsection{Differentiable Predictive Control}
\label{subsec:dpc_prelim}

{DPC} learns an explicit feedback law offline by treating the constrained
optimal-control problem as an unsupervised learning problem
\cite{Drgona2022,Drgona2024}. The prediction model is held fixed while the
policy parameters are optimized through finite-horizon closed-loop
rollouts. Let $F$ denote the discrete-time prediction map induced by the
model and the integration scheme, and let $\psi_t\in\Psi^{\mathrm{ctrl}}$
denote the control features available at time $t$, such as reference signals
or admissible bounds. The policy parameters ${w}$ are obtained by solving
{\begin{subequations}\label{eq:dpc_generic}
\begin{align}
\min_{{w}}\quad
& \frac{1}{m}\sum_{j=1}^{m}\Big(\sum_{t=1}^{N}
\ell_{\mathrm{DPC}}(x_t,u_{t-1},\psi_t)\notag\\
&\hspace{2.6cm}+\ell_{N}(x_N,\psi_N)\Big)
\label{eq:dpc_generic_a}\\
\text{s.t.}\quad
& x_{t+1}=F(x_t,u_t),\qquad u_t=\pi_{{w}}(x_t,\psi_t),
\label{eq:dpc_generic_b}\\
& x_t\in\mathcal{X},\quad u_t\in\mathcal{U},\quad \psi_t\in\Psi^{\mathrm{ctrl}},
\label{eq:dpc_generic_d}
\end{align}
\end{subequations}}
over $m$ sampled initial conditions and reference realizations{, with
$t=0,\ldots,N-1$ in \eqref{eq:dpc_generic_b}}. The stage loss
$\ell_{\mathrm{DPC}}$ collects the tracking objective, control regularization,
and differentiable penalties on constraint violation{; in practice, the
constraints \eqref{eq:dpc_generic_d} are enforced only through these
penalties}. {It may also penalize the increments $x_t-x_{t-1}$ and
$u_{t-1}-u_{t-2}$, with the convention $u_{-1}:=u_0$, and $\ell_N$ is an
optional terminal loss.} Since the rollout is differentiable in ${w}$, gradients
propagate through the entire horizon, so each action is optimized with its
downstream effect taken into account. Formulation \eqref{eq:dpc_generic} leaves
open how $F$ is obtained and how $\pi_{{w}}$ is parameterized;
Section~\ref{sec:sparsedpc} addresses both.

%%%%%%%%%%%%%%%%%%%%%%%%%%%%%%%%%%%%%%%%%%%%%%%%%%%%%%%%%%%%%%%%%%%%%%%%%%%%%%%%
% METHOD
%%%%%%%%%%%%%%%%%%%%%%%%%%%%%%%%%%%%%%%%%%%%%%%%%%%%%%%%%%%%%%%%%%%%%%%%%%%%%%%%

\section{Sparse Dictionary Differentiable Predictive Control}
\label{sec:sparsedpc}

SD-DPC proceeds in two stages, both instances of the thresholded
gradient-descent procedure of Algorithm~\ref{alg:TGD} applied to different
losses. A sparse dynamics model is first identified from multistep rollout
error, providing the differentiable prediction map; a sparse dictionary policy
is then trained by differentiating a finite-horizon closed-loop objective
through that map. Fig.~\ref{fig:sindy-dpc} summarizes the pipeline.

\subsection{Rollout-Based Sparse Identification}
\label{subsec:rollout_sindy}

We estimate $\Xi$ from rollout errors of the SINDy model rather than from
state-derivative data, so that $\dot{X}$ never enters the objective and the
identification does not inherit numerical differentiation error. Introducing
the one-step map induced by a fixed-step numerical integrator with step
$\Delta t$ and zero-order hold on the input,
\begin{equation}
x_{k+1}=\mathrm{ODESolve}\big({\hat f},\ x_k,\ u_k,\ \Delta t\big),
\label{eq:odesolve}
\end{equation}
{which we implement with the fourth-order Runge--Kutta scheme for the
continuous-time examples and use} consistently during identification and policy training. The data consist
of $m_{\mathrm{id}}$ trajectories
$\{(\hat{x}^j_0,\hat{u}^j_0),\ldots,(\hat{x}^j_{N_{\mathrm{id}}},\hat{u}^j_{N_{\mathrm{id}}})\}_{j=1}^{m_{\mathrm{id}}}$
collected under {persistently exciting} inputs, with $N_{\mathrm{id}}$ the
identification horizon. Applying the measured inputs to the model yields the
predicted states $\{x^j_k\}$, and the coefficients are obtained from
\begin{subequations}\label{eq:sindy-opt}
\begin{align}
\min_{\Xi}\quad
& \frac{1}{m_{\mathrm{id}}}\sum_{j=1}^{m_{\mathrm{id}}}\bigg(
    Q_1\big\|x_1^j-\hat{x}_1^j\big\|_2^2
    + Q_N\sum_{k=2}^{N_{\mathrm{id}}}\big\|x_k^j-\hat{x}_k^j\big\|_2^2\bigg)\notag\\
&\hspace{3.2cm}+\ \lambda\|\Xi\|_{1}
\label{eq:sindy-opta}\\
\text{s.t.}\quad
& x_{k+1}^j=\mathrm{ODESolve}\big({\hat f},\ x_k^j,\ \hat{u}_k^j,\ \Delta t\big),
\label{eq:sindy-optb}\\
& x_0^j=\hat{x}_0^j,
\label{eq:sindy-optc}
\end{align}
\end{subequations}
where $\lambda>0$ is the sparsity weight and $Q_1,Q_N\ge0$ weight the one-step
and multistep prediction errors.
With {$N_{\mathrm{id}}=1$ the multistep sum is empty, and} an explicit
Euler step {gives $x_1^j=\hat{x}_0^j+\Delta t\,\hat f(\hat{x}_0^j,\hat{u}_0^j)$,
which} is linear in $\Xi$, so \eqref{eq:sindy-opt} reduces to the $\ell_1$-regularized SINDy regression with forward-difference derivatives~\cite{Brunton2016a}; the
multistep terms extend it by penalizing the error accumulated along the rollout.

\subsection{Differentiable Sparse Dictionary Policy Parameterization}
\label{subsec:sparse_policy}

The feedback policy uses the same dictionary structure. With
$\mathcal{X}\subset\mathbb{R}^{n_x}$ and $\mathcal{U}\subset\mathbb{R}^{n_u}$
the admissible state and input sets, the policy input is $z_t=[x_t\ \psi_t]$
and the policy library is
\begin{equation}
\Theta_{\pi}(z)=\big[1,\ z^\top,\ \phi_{\mathrm{pol}}(z)^\top,\ \phi_{\mathrm{tr}}(z)^\top\big]\in\mathbb{R}^{1\times p_{\pi}},
\label{eq:thetaC}
\end{equation}
where $\phi_{\mathrm{pol}}$ collects quadratic and cross terms and
$\phi_{\mathrm{tr}}$ trigonometric terms, with maximal polynomial degree and
trigonometric frequency treated as hyperparameters. The sparse dictionary
policy is
\begin{equation}
u_t=\pi_{\Xi_{\pi}}(z_t)=\Pi_{\mathcal{U}}\big(\Theta_{\pi}(z_t)\,\Xi_{\pi}\big),
\qquad \Xi_{\pi}\in\mathbb{R}^{p_{\pi}\times n_u},
\label{eq:sparse_policy}
\end{equation}
with $\Pi_{\mathcal{U}}$ the componentwise projection enforcing
$\underline{u}\le u_t\le\overline{u}$. Column $i$ of $\Xi_\pi$ feeds input
channel $i$ alone, so sparsity in $\Xi_\pi$ determines directly which feedback
terms {are retained} in each channel, and the policy is evaluated without any
online optimization.
The exact derivative of the piecewise-affine projection $\Pi_{\mathcal U}$
annihilates the sensitivity of every saturated channel, so a differentiation
surrogate is required. With the raw channel output
$v_i(z):=[\Theta_\pi(z)\Xi_\pi]_i$, a band half-width $b_i\ge0$, a leak
$\epsilon\in[0,1]$, and the band
$\mathcal V_i:=[\underline u_i-b_i,\ \overline u_i+b_i]$, the forward pass
evaluates \eqref{eq:sparse_policy} exactly, so every loss and reported
trajectory is that of the true saturated closed loop, while the backward pass
uses the surrogate derivative
\begin{equation}
{\Gamma^{\epsilon}}:=\operatorname{diag}({\gamma^{\epsilon}_{1}},\dots,{\gamma^{\epsilon}_{n_u}}),
\qquad
{\gamma^{\epsilon}_{i}}:=
\begin{cases}
1, & v_i\in\mathcal V_i,\\[1pt]
\epsilon, & v_i\notin\mathcal V_i ,
\end{cases}
\label{eq:leaky_lambda}
\end{equation}
which recovers the {a.e.\ derivative of $\Pi_{\mathcal U}$} at $(\epsilon,b)=(0,0)$ and the plain
straight-through estimator~\cite{Bengio2013} at $(\epsilon,b)=(1,0)$. The
intermediate regime avoids two failure modes. With $\epsilon=0$, a coefficient
feeding a saturated channel receives no gradient {from the loss and only
shrinks}, so thresholding acts on stale rather than fitted evidence; with
$\epsilon=1$, a deeply saturated channel exerts influence through a
sensitivity the plant does not possess. The leak attenuates channels far
outside $\mathcal U$ without silencing them, while a small band $b$ keeps the
gradient from switching abruptly as a channel oscillates about its bound.

\subsection{SD-DPC Problem Formulation and Training Algorithm}

Adding a sparsity
penalty and using the sparse dictionary policy in \eqref{eq:dpc_generic} gives
{\begin{subequations}\label{eq:dpc}
\begin{align}
\min_{\Xi_{\pi}}\quad
& \frac{1}{m}\sum_{j=1}^{m}\Big(\sum_{t=1}^{N}
\ell_{\mathrm{DPC}}(x_t,u_{t-1},\psi_t)\notag\\
&\hspace{1.2cm}+\ell_{N}(x_N,\psi_N)\Big)
+\lambda_{\pi}\|\Xi_{\pi}\|_{1}
\label{eq:dpc_a}\\
\text{s.t.}\quad
& x_{t+1}=F_{\Xi}(x_t,u_t),\qquad u_t=\pi_{\Xi_{\pi}}(z_t),
\label{eq:dpc_b}\\
& x_0\sim\mathcal{P}_{x_0},\quad \psi_t\in\Psi^{\mathrm{ctrl}},
\label{eq:dpc_d}
\end{align}
\end{subequations}}
where $F_{\Xi}$ is the prediction map \eqref{eq:odesolve} of the identified
SINDy model. The input constraint $u_t\in\mathcal{U}$ holds by construction
through $\Pi_{\mathcal U}$, whereas the state constraint $x_t\in\mathcal{X}$
enters $\ell_{\mathrm{DPC}}$ through penalty terms, specified for each example
in Section~\ref{sec:experiments}.
Both problems have the composite form of a loss $\ell$, differentiable almost
everywhere, plus an entrywise $\ell_1$ penalty. Their coefficients are updated
by a first-order method, using a subgradient for the $\ell_1$ term, and are
periodically pruned, as in sequentially thresholded gradient training of SINDy
models~\cite{Champion2019}{; the two stages differ only in the loss and in
the matrix being sparsified (Algorithm~\ref{alg:TGD})}.
The entrywise penalty $\|\Xi\|_1=\sum_{i,j}|\Xi_{ij}|$ shrinks uninformative
coefficients toward zero without setting them exactly to zero, so a separate
thresholding step is required. Every ${n_p}$ epochs the retained
coefficient set is updated entrywise,
\begin{equation}
{\mathcal{S}}=\{\,(i,j)\ :\ |\Xi_{ij}|\ge\tau\,\},
\label{eq:retained_set}
\end{equation}
so each output channel $j$ retains its own subset of dictionary terms.
Entries outside ${\mathcal{S}}$ are set to zero and excluded from
subsequent updates, a library column is dropped from the compiled evaluation
once its term is {removed from} every channel, and the optimizer is
reinitialized on the reduced parameter set. The procedure stops once
${\mathcal{S}}$ has been stable over $M$ consecutive prunes or the epoch
budget is exhausted, after which the coefficients {in $\mathcal{S}$} are
refitted without the $\ell_1$ penalty to remove its shrinkage bias. In our implementation, the update in
line~{3} is carried out with AdamW \cite{loshchilov2017decoupled}; the
properties below concern the {pruning mechanism} and do not depend on this
choice.

\begin{algorithm}[!t]
\caption{Sequentially Thresholded Gradient Descent with Debiasing}
\label{alg:TGD}
\small
\setlength{\abovedisplayskip}{2pt}
\setlength{\belowdisplayskip}{2pt}
\setlength{\jot}{2pt}
\begin{algorithmic}[1]
\Require $\Theta$: library of $p$ terms;\ \ $\ell$: loss (a.e.\ differentiable)
\Statex \hspace{\algorithmicindent}$\Xi\in\mathbb{R}^{p\times n}$: coefficients,
$\Xi_{ij}$ = term $i$, channel $j$
\Statex \hspace{\algorithmicindent}$\tau$: threshold;\ ${n_p}$: prune interval;\
$n_e$: epochs;\ $M$: patience;\ $K$: debias epochs
\Statex \textbf{Note:} applied to \eqref{eq:sindy-opt} for $\Xi$
($n=n_x$), and to \eqref{eq:dpc} for $\Xi_\pi$ ($n=n_u${, $(\lambda,\tau)\to(\lambda_\pi,\tau_\pi)$}) with
$\partial\Pi_{\mathcal U}={\Gamma^{\epsilon}}$ of \eqref{eq:leaky_lambda}.
\State ${\mathcal{S}_{\mathrm{prev}}} \gets$ all entries of $\Xi$;\quad $s \gets 0$
\For{$e=1,\dots,n_e$}
\State Update $\Xi_{{\mathcal{S}_{\mathrm{prev}}}}$ by minimizing $\ell+\lambda\|\Xi\|_1$
\If{$e \bmod {n_p} = 0$}
\State ${\mathcal{S}} \gets \{(i,j):|\Xi_{ij}|\ge\tau\}$
\Comment{per-channel}
\If{${\mathcal{S}\neq\mathcal{S}_{\mathrm{prev}}}$}
\State $\Xi_{ij}\gets 0$ for all $(i,j)\notin{\mathcal{S}}$
\State Drop from $\Theta$ the terms {removed from} every channel
\State Reinitialize optimizer on the reduced set
\State ${\mathcal{S}_{\mathrm{prev}}\gets\mathcal{S}}$;\quad $s\gets 0$
\Else
\State $s\gets s+1$;\quad \textbf{if} $s\ge M$ \textbf{then break}
\EndIf
\EndIf
\EndFor
\State Save the biased sparse coefficients $\Xi^{\mathrm{bias}}\gets\Xi$
\State Warm-start $\Xi_{{\mathcal{S}}}$ from $\Xi^{\mathrm{bias}}$ and reinitialize the optimizer
\For{$e=1,\dots,K$}
\State Update only $\Xi_{{\mathcal{S}}}$ by minimizing $\ell$ with $\lambda=0$
\Comment{fixed {retained set}}
\EndFor
\State \Return debiased sparse coefficient matrix $\Xi$
\end{algorithmic}
\end{algorithm}

\vspace{-0.1in}
\begin{proposition}[Finite support stabilization]
\label{prop:support}
Let $\tau>0$. In Algorithm~\ref{alg:TGD}, a coefficient that has been pruned is
never retained again. Consequently, the retained set changes at most $pn$
times, and the thresholding loop terminates after at most $n_pM(pn+1)$ epochs,
for any loss and any epoch budget $n_e$.
\end{proposition}

\begin{proof}
A pruned coefficient is set to zero and excluded from all subsequent updates,
so its value remains zero. Since $\tau>0$, it fails the test
$|\Xi_{ij}|\ge\tau$ at every later prune and is never retained again. Hence,
if $\mathcal{S}_r$ denotes the retained set after the $r$-th change, each
change removes at least one entry, and
$\mathcal{S}_{r+1}\subsetneq\mathcal{S}_r$. As $\Xi$ has $pn$ entries, at most
$pn$ such changes can occur. The loop stops after $M$ consecutive prunes
without a change, so each change is preceded by at most $M-1$ unchanged
prunes, and the last change is followed by exactly $M$. The total number of
prunes is therefore at most $M(pn+1)$, and since prunes are $n_p$ epochs
apart, the loop terminates within $n_pM(pn+1)$ epochs. {\unskip\nobreak\hfill$\square$}
\end{proof}
\vspace{0.08in}
Proposition~\ref{prop:support} splits Algorithm~\ref{alg:TGD} into a finite
phase that selects which terms are retained and a refit that only adjusts
their values. {Since pruning occurs before the optimization has
converged, its cost is bounded through a stationarity residual.}
\vspace{-0.25in}
{\begin{proposition}[Pruning cost]\label{prop:prune}
Let $J=\ell+\lambda\|\cdot\|_1$, let $\Xi'$ be obtained from $\Xi$ by
zeroing a set $\mathcal D$ of $q$ retained entries with $0<|\Xi_{ij}|<\tau$,
let $\nabla\ell$ be $\beta$-Lipschitz on the segment between them, and let
the residuals
$r_{ij}=\partial\ell(\Xi)/\partial\Xi_{ij}+\lambda\,\mathrm{sign}(\Xi_{ij})$
on $\mathcal D$ satisfy $\|r_{\mathcal D}\|_F\le\eta$. Then
$J(\Xi')\le J(\Xi)+\eta\tau\sqrt q+\tfrac{\beta}{2}q\tau^2$ and
$\ell(\Xi')\le\ell(\Xi)+\lambda q\tau+\eta\tau\sqrt q+\tfrac{\beta}{2}q\tau^2$.
\end{proposition}
\begin{proof}
Let $\Delta=\Xi'-\Xi$, so that $\|\Delta\|_F\le\tau\sqrt q$. The descent lemma
and $\lambda\|\Xi'\|_1-\lambda\|\Xi\|_1=\langle\lambda\,\mathrm{sign}(\Xi),\Delta\rangle$
give $J(\Xi')-J(\Xi)\le\langle r_{\mathcal D},\Delta_{\mathcal D}\rangle
+\tfrac{\beta}{2}\|\Delta\|_F^2$, and the Cauchy--Schwarz inequality yields
the first bound; the second adds
$\lambda\sum_{\mathcal D}|\Xi_{ij}|\le\lambda q\tau $. {\unskip\nobreak\hfill$\square$}
\end{proof}}

\vspace{0.1in}
{The residual vanishes at a stationary point and otherwise absorbs the
optimizer's inexactness and the surrogate bias \eqref{eq:leaky_lambda}; the
bound on $\ell$ shows why $\lambda$ and $\tau$ are tuned jointly
(Fig.~\ref{fig:lambda-sparsity}). In the policy stage, smoothness holds along
segments on which no channel crosses a saturation or constraint boundary.}

\subsection{\texorpdfstring{{Structural Properties of SD-DPC Closed Loops}}{Structural Properties of SD-DPC Closed Loops}}
\label{subsec:properties}
{Since the learned policy is a finite sum of known functions, its
sensitivity can be bounded explicitly from the retained terms. We write $\pi=\Pi_{\mathcal U}(v)$ with
the raw output $v(x,\psi)=\Theta_\pi(x,\psi)\Xi_\pi$, let $\mathcal S_j$
collect the terms retained in channel $j$, and let $\|\cdot\|$ denote the
Euclidean norm.}
{\begin{assumption}\label{ass:lip}
States and features lie in compact sets $\tilde{\mathcal X}$ and
$\tilde\Psi$, on which (a)~each library function $\theta_{\pi,i}(\cdot,\psi)$
is $L_i$-Lipschitz for every $\psi\in\tilde\Psi$; (b)~$\|\Theta_\pi(x,\psi)\|\le
B_\Theta$; and (c)~$\|F(x,u)-F(x',u')\|\le L_x\|x-x'\|+L_u\|u-u'\|$ for all
$x,x'\in\tilde{\mathcal X}$ and $u,u'\in\mathcal U$.
\end{assumption}}
\vspace{-0.2in}
{\begin{proposition}[Lipschitz bound and perturbation of sparse policies]\label{prop:lipschitz}
Let Assumption~\ref{ass:lip} hold. Then, for all $x,x'\in\tilde{\mathcal X}$
and $\psi\in\tilde\Psi$:
\begin{enumerate}
\item[(i)] \emph{Lipschitz bound:}
$\|\pi(x,\psi)-\pi(x',\psi)\|\le L_\pi\|x-x'\|$, with
\begin{equation}
L_\pi:=\Big(\textstyle\sum_{j=1}^{n_u}\big(\sum_{i\in\mathcal S_j}|\Xi_{\pi,ij}|\,L_i\big)^2\Big)^{1/2}.
\label{eq:Lpi}
\end{equation}
\item[(ii)] \emph{Policy perturbation:} the policy
$\tilde\pi=\Pi_{\mathcal U}(\Theta_\pi\tilde\Xi_\pi)$ with perturbed
coefficients satisfies
$\|\pi(x,\psi)-\tilde\pi(x,\psi)\|\le B_\Theta\|\Xi_\pi-\tilde\Xi_\pi\|_F$;
in particular, pruning $q$ coefficients below $\tau_\pi$ changes the action
by at most $B_\Theta\tau_\pi\sqrt q$.
\item[(iii)] \emph{Trajectory perturbation:} let $\tilde\pi$ be any policy
with values in $\mathcal U$ and $\|\pi-\tilde\pi\|\le\varepsilon_\pi$ on
$\tilde{\mathcal X}\times\tilde\Psi$, let $\|F-\tilde F\|\le\delta$ on
$\tilde{\mathcal X}\times\mathcal U$, and let
$x_{t+1}=F(x_t,\pi(x_t,\psi_t))$ and
$\tilde x_{t+1}=\tilde F(\tilde x_t,\tilde\pi(\tilde x_t,\psi_t))$ start from
the same $x_0$ under the same $\psi_t\in\tilde\Psi$ and remain in
$\tilde{\mathcal X}$ for $t\le N$. Then
\begin{equation}
\|x_t-\tilde x_t\|\le e_t:=(\delta+L_u\varepsilon_\pi)\textstyle\sum_{s=0}^{t-1}a^s
\label{eq:deviation}
\end{equation}
for $t\le N$, where $a:=L_x+L_uL_\pi$. Moreover, if $h$ is $L_h$-Lipschitz on
$\tilde{\mathcal X}$ and one of the two trajectories satisfies $h\ge L_he_t$
for $t\le N$, the other satisfies $h\ge0$.
\end{enumerate}
\end{proposition}}
{\begin{proof}
(i) The componentwise projection $\Pi_{\mathcal U}$ is $1$-Lipschitz, so
$\|\pi(x,\psi)-\pi(x',\psi)\|\le\|v(x,\psi)-v(x',\psi)\|$. By
Assumption~\ref{ass:lip}(a),
$|v_j(x,\psi)-v_j(x',\psi)|\le\sum_{i\in\mathcal S_j}|\Xi_{\pi,ij}|\,
|\theta_{\pi,i}(x,\psi)-\theta_{\pi,i}(x',\psi)|\le\big(\sum_{i\in\mathcal S_j}
|\Xi_{\pi,ij}|L_i\big)\|x-x'\|$, and summing the squares over $j$ gives (i).
(ii) Likewise,
$\|\pi-\tilde\pi\|\le\|\Theta_\pi(\Xi_\pi-\tilde\Xi_\pi)\|\le\|\Theta_\pi\|\,
\|\Xi_\pi-\tilde\Xi_\pi\|_F$, and $q$ removed entries below $\tau_\pi$ have
Frobenius norm at most $\tau_\pi\sqrt q$. (iii) With $d_t:=\|x_t-\tilde x_t\|$
and $\psi_t$ omitted, split $x_{t+1}-\tilde x_{t+1}$ into
$F(x_t,\pi(x_t))-F(\tilde x_t,\pi(\tilde x_t))$,
$F(\tilde x_t,\pi(\tilde x_t))-F(\tilde x_t,\tilde\pi(\tilde x_t))$, and
$F(\tilde x_t,\tilde\pi(\tilde x_t))-\tilde F(\tilde x_t,\tilde\pi(\tilde x_t))$.
By Assumption~\ref{ass:lip}(c) and (i), these are bounded by $a\,d_t$,
$L_u\varepsilon_\pi$, and $\delta$, so $d_{t+1}\le a\,d_t+L_u\varepsilon_\pi+\delta$
with $d_0=0$, and induction gives \eqref{eq:deviation}. Finally,
$|h(x_t)-h(\tilde x_t)|\le L_hd_t\le L_he_t$. {\unskip\nobreak\hfill$\square$}
\end{proof}}
\vspace{0.1in}
{All constants are explicit: $L_i=\omega$ for $\sin\omega x_k$ and
$\cos\omega x_k$, $L_i=0$ for terms in $\psi$ alone, and, on a convex
$\tilde{\mathcal X}$, $L_i$ is the largest gradient norm of $\theta_{\pi,i}$.
Since $L_\pi\le(\max_iL_i)\|\Xi_\pi\|_1$, the $\ell_1$ penalty also limits the
sensitivity. The model constants $L_x$ and $L_u$ follow from $\Xi$ and the
integrator in the same way, so Proposition~\ref{prop:lipschitz} is evaluated
from the two sparse coefficient matrices and the library, whereas for deep neural policies such
bounds require products of layer norms or semidefinite relaxations
\cite{fazlyab2019efficient}. By (ii) and (iii), a pruning step ($\tilde F=F$) or a
model error with $\tilde\pi=\pi$, such as the transfer from $F_\Xi$ to the
plant, moves the trajectory in proportion to its size, with growth at most
geometric at rate $a$; since typically $a>1$, the bound is conservative over
long horizons, and $\tilde{\mathcal X}$ is chosen to contain both
trajectories (Section~\ref{sec:results}).}

\section{Numerical Examples}
\label{sec:experiments}

Three numerical case studies demonstrate the
flexibility of SD-DPC: constrained reference tracking of a nonlinear system, {stabilization of a nonlinear oscillator at its unstable origin}, and an obstacle-avoidance reaching task. All examples are implemented in NeuroMANCER
v1.5.3~\cite{Neuromancer2023}, and all computations are performed on a MacBook
with an Apple M1 processor and 8\,GB of RAM. The code is available at
\url{https://github.com/SOLARIS-JHU/SD-DPC}.

In each case study, SD-DPC is compared against three baselines. The first
is an optimization-based benchmark solved with IPOPT
\cite{wachter2006implementation} through CasADi \cite{Andersson2019},
{which minimizes the same stage and terminal costs as the policy
objective{, unless stated otherwise in the case study, and uses the true
dynamics, and therefore has exact model knowledge}. For the two-tank system, it is a
receding-horizon MPC whose horizon starts at 50 steps, shrinks toward one
step as the end of each reference segment approaches, and is reset when the
reference changes. For the Van der Pol oscillator and the 2-D integrator,
the full $N$-step problem is solved once per scenario and the optimal input
sequence is applied in open loop. Because this benchmark is computed for
each scenario separately, whereas a learned policy is a single feedback law
shared by all scenarios, it serves as a performance reference rather than
as a competitor.}
NN-DPC is a deep neural policy trained by
\eqref{eq:dpc_generic} on the same identified model and represents the
standard DPC parameterization; in all cases, it is a multilayer perceptron
with two hidden layers of $32$ units and {GELU} activations. Finally,
NN-Distilled follows the distillation approach of SINDy-RL
\cite{Zolman2025}: a sparse dictionary policy is regressed onto the
state--action pairs generated by the trained NN-DPC policy. Its terms are
fixed to those retained by SD-DPC, so the two differ only in how the
coefficients on this common support are obtained, either by regression onto
the neural policy's actions or under the closed-loop objective
\eqref{eq:dpc}{, which isolates the effect of the training objective}.

{All policies are trained on
rollouts of the identified model $F_\Xi$, whereas all controllers are
evaluated on the true plants \eqref{eq:twotank}, \eqref{eq:vdp}, and
\eqref{eq:dint}, so that the reported results include the effect of model
mismatch.}
{Since state constraints, including the obstacle, are only penalized,
their satisfaction is assessed on the evaluation scenarios of
Section~\ref{sec:results}.}

%--------------------------------------------------------------------------\

\subsection{Two-Tank System}
\label{subsec:twotank}

We consider a cascaded two-tank system in which a pump feeds an upper tank that
drains by gravity into a lower one. The state $x=[x_1,x_2]^\top$ collects the
two liquid levels and the inputs $u=[u_1,u_2]^\top$ are the
pump command and the valve split,
\begin{subequations}\label{eq:twotank}
\begin{align}
\dot{x}_{1} &= c_1(1 - u_{2})u_{1} - c_2\sqrt{x_{1}}, \label{eq:twotank_a}\\
\dot{x}_{2} &= c_1 u_{2}u_{1} + c_2\sqrt{x_{1}} - c_2\sqrt{x_{2}},
\label{eq:twotank_b}
\end{align}
\end{subequations}
with inlet and outlet valve coefficients $(c_1,c_2)=(0.08,0.04)$.
{The policy drives the levels to time-varying references $r_t$ within the
level bounds $\underline{x}=0$ and $\overline{x}=1$, with inputs in
$\mathcal U=[0,1]^2$ and features $\psi_t=(r_t,\underline x_t,\overline x_t)$.
The stage loss combines reference tracking with penalties on the level
bounds,}
{\begin{equation}\label{eq:dpc_loss}
\begin{aligned}
\ell_{\mathrm{DPC}}&=Q_r\|r_t-x_t\|_2^2+\ell_{\underline{x},t}+\ell_{\overline{x},t},\\
\ell_{\underline{x},t}&=Q_x\big\|[\underline{x}-x_t]_+\big\|_2^2,\quad
\ell_{\overline{x},t}=Q_x\big\|[x_t-\overline{x}]_+\big\|_2^2,
\end{aligned}
\end{equation}}
{where $[z]_+=\max(0,z)$, and there is no terminal loss, $\ell_N\equiv0$.}

%--------------------------------------------------------------------------
\subsection{Forced Van der Pol Oscillator}
\label{subsec:vdp}

We next consider the forced Van der Pol oscillator, whose uncontrolled
trajectories converge to an attracting limit cycle surrounding an unstable
origin. With $x=[x_1,x_2]^\top$ the position and velocity and $u$ the forcing,
\begin{equation}\label{eq:vdp}
\dot{x}_{1} = x_{2},\qquad
\dot{x}_{2} = \mu (1 - x_{1}^2) x_{2} - x_{1} + u,\qquad \mu = 1.
\end{equation}
{The policy regulates the state to the origin, thereby suppressing the
limit cycle, with input bounds $\overline{u}=-\underline{u}=5$ and state
bounds $\overline{x}=-\underline{x}=[5,5]^\top$. Since the target is the
fixed origin rather than an exogenous signal, the policy is pure state
feedback, $z_t=x_t$. The stage loss penalizes the state, its increments, and
the input increments, together with the bound penalties of
\eqref{eq:dpc_loss},}
{\begin{equation}\label{eq:loss_vdp}
\begin{aligned}
\ell_{\mathrm{DPC}}={}&Q_r\|x_t\|_2^2
+Q_u\|u_{t-1}\|_2^2\\
&+Q_{du}\|u_{t-1}-u_{t-2}\|_2^2
+\ell_{\underline{x},t}+\ell_{\overline{x},t}.
\end{aligned}
\end{equation}}
{and $\ell_N\equiv0$.}
Initial conditions are sampled both on and off the limit cycle so that the
coefficients are not specialized to one region of the phase plane.

\begin{figure}[t]
    \centering
    \includegraphics[width=0.85\columnwidth]{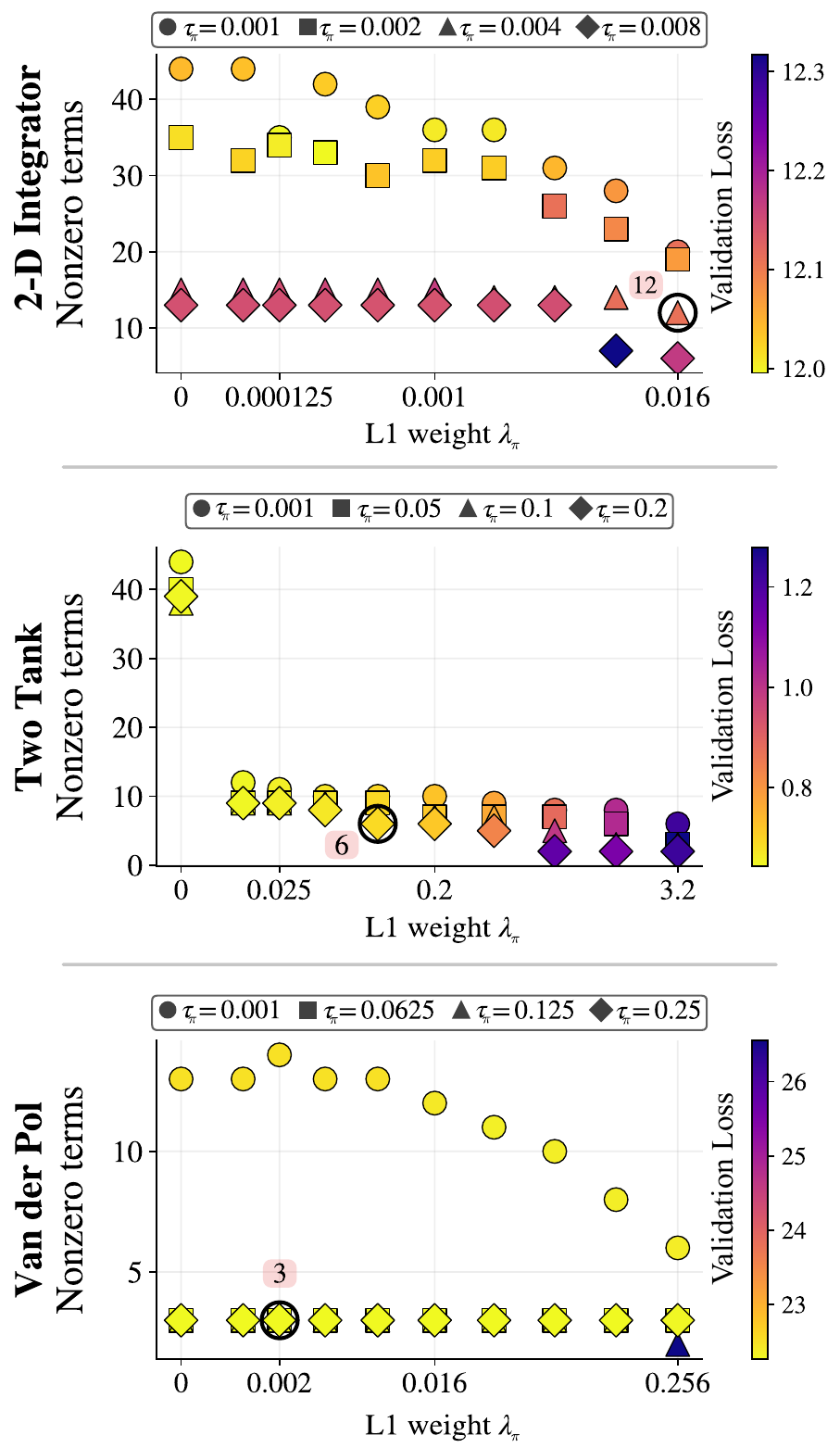}
    \caption{Policy sparsity versus $\ell_1$ weight $\lambda_{\pi}$. Marker
    shape denotes pruning threshold $\tau_\pi$, color shows the loss (excluding the $\ell_1$ penalty) on the
    validation scenarios used for selection, and black circles indicate the selected
    policies.}
    \label{fig:lambda-sparsity}
\end{figure}

\begin{figure*}[t]
    \centering
    \includegraphics[width=\textwidth]{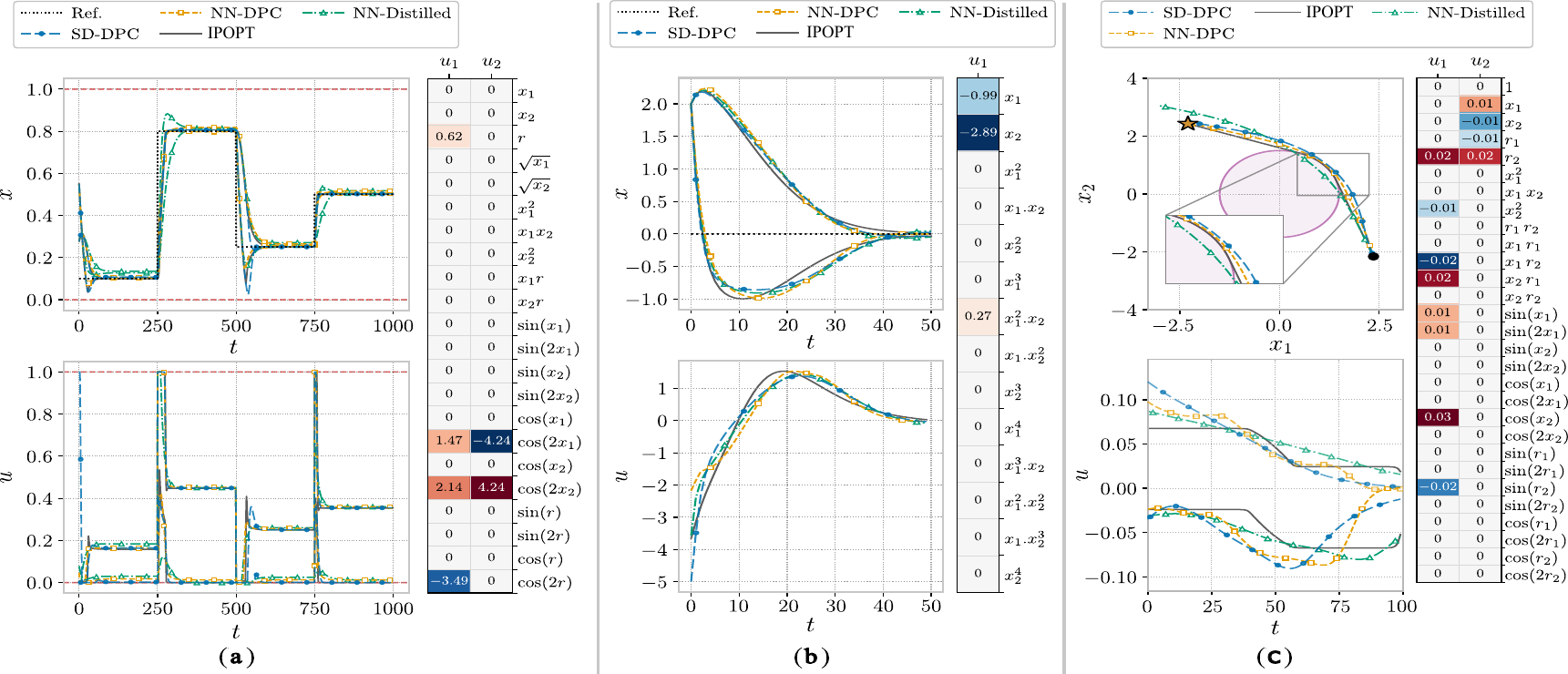}
    \caption{Closed-loop trajectories of SD-DPC, NN-DPC,
    NN-Distilled, and the IPOPT benchmark, with the retained SD-DPC policy coefficients, for (a)~the two-tank
    system, (b)~the forced Van der Pol oscillator, and (c)~the 2-D integrator
    with obstacle avoidance.}
    \label{fig:closed-loop-results}
\end{figure*}

%--------------------------------------------------------------------------
\subsection{Two-Dimensional Integrator with Obstacle Avoidance}
\label{subsec:obstacle}

The third case study demonstrates handling of nonlinear state constraints. We
consider a discrete-time integrator in the plane, in which each input sets the
increment of the corresponding position coordinate,
\begin{equation}\label{eq:dint}
x_{{t}+1}=A\,x_{{t}}+B\,u_{{t}},
\qquad
A=\begin{bmatrix} 1 & 0\\ 0 & 1\end{bmatrix},
\quad
B=\begin{bmatrix} 1 & 0\\ 0 & 1\end{bmatrix},
\end{equation}
with $x\in\mathbb R^2$ and $u\in\mathbb R^2$. The dynamics are linear and
decoupled, so the difficulty of this example lies entirely in the constraint,
which excludes a region of the state space,
\begin{equation}\label{eq:obstacle}
h_{\mathrm{obs}}(x)=b_o\,(x_{1}-{x^{o}_{1}})^2+(x_{2}-{x^{o}_{2}})^2-{(p_o/2)^2}\ \ge\ 0,
\end{equation}
where $p_o$, $b_o$, and {$x^{o}=(x^{o}_{1},x^{o}_{2})$} set the size, shape and center of the obstacle{, here $p_o=3$, $b_o=1$, and $x^{o}=(0,0)$, i.e., a circle of radius $1.5$ centered at the origin}.
The objective is to reach a terminal setpoint $r_N$ from initial conditions
placed so that the direct path is blocked{, with inputs bounded by
$\underline u\le u_t\le\overline u$. The stage and terminal losses are}
{\begin{equation}\label{eq:loss_obs}
\begin{aligned}
\ell_{\mathrm{DPC}}={}&Q_{dx}\|x_t-x_{t-1}\|^2+Q_u\|u_{t-1}\|^2\\
&+Q_{du}\|u_{t-1}-u_{t-2}\|^2+Q_h\big[-h_{\mathrm{obs}}(x_t)\big]_+,\\
\ell_N={}&Q_r\|x_N-r_N\|^2+Q_{u_f}\|\pi_{\Xi_\pi}(r_N,r_N)\|^2 .
\end{aligned}
\end{equation}}
The obstacle enters through the one-sided hinge
$[-h_{\mathrm{obs}}(x_t)]_+$, so the penalty grows linearly in the depth of
violation. The terminal term $Q_{u_f}\|\pi_{\Xi_\pi}(r_N,r_N)\|^2${,
which evaluates the policy at the target and therefore depends on $\Xi_\pi$
directly,} is required
because without it the policy holds a residual increment at the target and the
reaching task does not settle. {For this example, the IPOPT benchmark
instead minimizes the control effort $\sum_t\|u_t\|_1$ subject to the
terminal condition $x_N=r_N$ and the obstacle constraint, both imposed as hard
constraints, so that it reaches the setpoint exactly by construction.}

\begin{table}[!t]
\centering
\caption{Identification and policy training parameters.}
\label{tab:hyperparameters}
\footnotesize
\setlength{\tabcolsep}{2pt}
\renewcommand{\arraystretch}{1.05}
\begin{tabular*}{\columnwidth}{@{\extracolsep{\fill}} l r r r @{}}
\toprule
& \multicolumn{1}{c}{Two-Tank}
& \multicolumn{1}{c}{Van der Pol}
& \multicolumn{1}{c}{2-D Int.}\\
\midrule
\multicolumn{4}{@{}l}{\textbf{Identification}} \\[-1pt]
Trajectories $m_{\mathrm{id}}$   & $2000$              & $2000$            & $1500$ \\
Horizon $N_{\mathrm{id}}$        & $5$                 & $3$               & $3$ \\
Integrator                       & RK4                 & RK4               & discrete \\
Step $\Delta t$ [s]              & $1.0$               & $0.1$             & -- \\
Epochs $n_e$                     & $20\,000$           & $6000$            & $20\,000$ \\
Weights $(Q_1,Q_N)$              & $(1,2)$             & $(1,2)$           & $(1,2)$ \\
Sparsity $\lambda$               & $3\times10^{-6}$    & $1\times10^{-5}$  & $1\times10^{-9}$ \\
Threshold $\tau$                 & $5\times10^{-3}$    & $1\times10^{-2}$  & $1\times10^{-2}$ \\
Prune interval $n_p$             & $5000$              & $1000$            & $5000$ \\
Library size $p$                 & $15$                & $18$              & $19$ \\
\addlinespace[3pt]
\multicolumn{4}{@{}l}{\textbf{Policy}} \\[-1pt]
Scenarios $m$                    & ${2000}$        & $1000$            & $2000$ \\
Horizon $N$                      & $50$                & $50$              & $100$ \\
Sparsity $\lambda_\pi$           & $0.1$               & ${0.001}$     & $0.016$ \\
Threshold $\tau_\pi$             & $0.2$               & ${0.0625}$    & $0.004$ \\
Library size $p_\pi$             & $22$                & $14$              & $29$ \\
Leak $\epsilon$                  & ${0.01}$        & ${0.01}$      & ${0.05}$ \\
Band $b$                         & ${0.05}$        & ${0.05}$      & ${0.01}$ \\
Stage weights
& \makecell[r]{$Q_r={20}$\\$Q_x={20}$}
& \makecell[r]{$Q_r=20$, $Q_x=20$\\${Q_u=5}$, $Q_{du}=1$}
& \makecell[r]{$Q_r=1$, $Q_{dx}=1$\\$Q_u=Q_{du}=10$\\$Q_h=100$, $Q_{u_f}=5$}\\
\bottomrule
\end{tabular*}
\begin{minipage}{\columnwidth}
\footnotesize
\vspace{3pt}
\textit{Notes:} {The policy epoch budget is $12\,000$,} with the
{retained set pruned every $n_p=800$} epochs. Patience is $M=3$ throughout. Both stages use
AdamW \cite{loshchilov2017decoupled} under Algorithm~\ref{alg:TGD} with weight
decay $0.01${, which remains active during the debiasing refit of
$K={800}$ epochs.}
\end{minipage}
\end{table}
%--------------------------------------------------------------------------
\raggedbottom
\subsection{Training Setup}
\label{subsec:training}

In each case the prediction map $F_\Xi$ is identified first by solving
\eqref{eq:sindy-opt} and then held fixed while the policy coefficients are
learned; both stages are sparsified by Algorithm~\ref{alg:TGD}. The constants
for both stages are collected in Table~\ref{tab:hyperparameters}.
{The retained terms depend jointly on $\lambda_\pi$ and $\tau_\pi$; we
select the pair by the} sweep of Fig.~\ref{fig:lambda-sparsity}, retaining the sparsest policy whose
task performance on a separate set of validation scenarios is not degraded
relative to the dense end of the sweep; the evaluation scenarios in
Section~\ref{sec:results} are held out from this selection. The price of sparsity is system dependent: for the Van der Pol oscillator the loss
is flat across the sweep, whereas for the two-tank system it degrades once
$\lambda_\pi$ is pushed beyond the selected point.
%--------------------------------------------------------------------------
%--------------------------------------------------------------------------

\section{Results}
\label{sec:results}

\begin{table}[!t]
\centering
\caption{Controller evaluation across case studies.}
\label{tab:combined-controller-evaluation}
\scriptsize
\setlength{\tabcolsep}{0pt}
\renewcommand{\arraystretch}{1.1}
\begin{tabular*}{\columnwidth}{@{\extracolsep{\fill}} l r @{\hspace{5pt}} c r @{\hspace{4pt}} r @{}}
\toprule
Method
& \makecell[c]{Task loss\\$[10^{-3}]$}
& \makecell[c]{Control\\effort}
& \makecell[r]{Time $[\mu\mathrm{s}]$\\/\,memory}
& \makecell[r]{Violations\\$[\%]\,/\,$depth}\\
\midrule
\multicolumn{5}{@{}l}{\textbf{\textit{Two-Tank}}}\\
IPOPT        & $8.8\!\times\!10^{-6}$     & --             & $8444\,/\,26$\,MiB                        & $0.0\,/\,\text{--}$\\
NN-DPC       & $0.144\pm0.013$            & --             & $30\,/\,5.01$\,kB                         & $0.0\,/\,\text{--}$\\
SD-DPC       & $\mathbf{0.012\pm0.006}$   & --             & $\mathbf{17\,/\,0.047}$\,\textbf{kB}      & $0.0\,/\,\text{--}$\\
NN-Distilled & $0.209\pm0.039$            & --             & $17\,/\,0.047$\,kB                        & $0.0\,/\,\text{--}$\\
\addlinespace[2pt]
\multicolumn{5}{@{}l}{\textbf{\textit{Van der Pol}}}\\
IPOPT        & $4.6\!\times\!10^{-24}$    & --             & $301\,/\,19$\,MiB                         & $0.0\,/\,\text{--}$\\
NN-DPC       & $\mathbf{41.0\pm124.4}$    & --             & $29\,/\,5.01$\,kB                         & $0.0\,/\,\text{--}$\\
SD-DPC       & $49.9\pm126.4$             & --             & $\mathbf{17\,/\,0.023}$\,\textbf{kB}      & $0.0\,/\,\text{--}$\\
NN-Distilled & $96.8\pm183.8$             & --             & $17\,/\,0.023$\,kB                        & $0.0\,/\,\text{--}$\\
\addlinespace[2pt]
\multicolumn{5}{@{}l}{\textbf{\textit{2-D Integrator}}}\\
IPOPT        & $0.000\pm0.000$            & $12.00\pm1.18$ & $1841\,/\,18$\,MiB                        & $0.0\,/\,\text{--}$\\
NN-DPC       & $\mathbf{0.298\pm0.560}$   & $\mathbf{12.00\pm1.16}$ & $28\,/\,5.01$\,kB                & $0.0\,/\,\text{--}$\\
SD-DPC       & $4.75\pm4.79$              & $12.07\pm1.12$ & $\mathbf{17\,/\,0.094}$\,\textbf{kB}      & {$0.2\,/\,0.01$}\\
NN-Distilled & $62.6\pm58.0$              & $12.07\pm1.03$ & $17\,/\,0.094$\,kB                        & $60.8\,/\,1.61$\\
\bottomrule
\end{tabular*}
\par\vspace{3pt}
\parbox{\columnwidth}{\footnotesize\textit{Notes:} Mean $\pm$ standard
deviation over 500 held-out test scenarios simulated on the true plant; IPOPT
losses are numerically zero and given as means. Horizons are 1000, 50, and
100 steps; task loss is the MSE over the final 100 steps of each 250-step
reference segment, the final 10 steps, and the terminal state, respectively.
Control effort is $\sum_t\lVert u_t\rVert_1$. Time is per control step; for
IPOPT it is the solve time divided by $N$. Memory is the parameter storage for
the learned policies, and the resident memory added by each solve for IPOPT. Violations give the percentage of scenarios
with any state-constraint violation and the maximum violation depth, measured
by $[-h_{\mathrm{obs}}(x_t)]_+$. NN-Distilled uses the terms retained by
SD-DPC, fitted by unregularized least squares.}
\end{table}
Fig.~\ref{fig:closed-loop-results} and
Table~\ref{tab:combined-controller-evaluation} summarize the closed-loop
evaluation. The identified models used for policy training are accurate over the
prediction horizon: on held-out trajectories, their $N=50$ step rollout
mean-squared errors are $7.64\times10^{-10}$, $3.30\times10^{-8}$, and
$4.71\times10^{-7}$ for the two-tank, Van der Pol, and 2-D integrator systems,
respectively, with a maximum absolute state error below
$1.37\times10^{-3}$.
Across all three systems, the sparse policy evaluates in
$17\,\mu\mathrm{s}$ and occupies under $0.1\,\mathrm{kB}$, roughly two orders of magnitude less memory than NN-DPC {and far
below the $18$--$26$\,MiB allocated by each IPOPT solve}. Its closed-loop behavior, however, differs from task to task.

{Proposition~\ref{prop:lipschitz} is evaluated on the levels
$x\in[0.05,1]^2$ and references $r\in[0.1,0.9]$ for the two-tank system and,
for the other two systems, on a box enclosing all evaluation trajectories.
Weighting each retained coefficient by the slope of its library function,
\eqref{eq:Lpi} certifies $L_\pi\le18.4$, $16.0$, and $0.32$, whereas the
largest $\|\partial_x\pi\|$, computed by automatic differentiation on a dense
grid over these sets, is $12.0$, $9.4$, and $0.19$, so the certificate is
conservative by a factor of at most $1.7$; the pruning bound (ii) also held
at every prune of Algorithm~\ref{alg:TGD}.}

\raggedbottom
\subsection{Two-Tank System}
For the two-tank system, SD-DPC retains six terms across the two input
channels (Fig.~\ref{fig:closed-loop-results}(a)), tracks each reference
segment, and satisfies the level constraints on every scenario. Its tracking loss is about twelve times lower than that of NN-DPC, a gap well
beyond the scenario-to-scenario spread of either policy, so the sparse
parameterization costs nothing in accuracy here; both learned policies remain
well above the {IPOPT} benchmark, which optimizes each scenario online. The
distilled policy, which uses the same six terms, also satisfies the
constraints, but its coefficients, fitted to the actions of NN-DPC, leave a steady offset, and its tracking loss is
about seventeen times that of SD-DPC.

The raw outputs of the learned SD-DPC pump and valve channels are

\begin{equation}
\begin{aligned}
v_1 &= 0.6198\,r + 1.4659\cos 2x_1 + 2.1415\cos 2x_2\\
&\quad - 3.4944\cos 2r,\\
v_2 &= -4.2401\cos 2x_1 + 4.2429\cos 2x_2,
\end{aligned}
\label{eq:twotank-policy}
\end{equation}
with $u=\Pi_{[0,1]^2}(v)$ and $r$ the level reference. The pump command
combines a reference-dependent feedforward with feedback on both levels,
while the two nearly equal valve coefficients make the valve open only when
the upper tank is fuller than the lower one.

\subsection{Forced Van der Pol Oscillator}
In the Van der Pol case, all controllers stabilize the origin and satisfy the
constraints on every scenario. The benchmark optimizes each scenario individually and therefore reaches
   a loss below $10^{-5}$, whereas both learned policies settle more slowly,
   with a final-10-step loss of about $4$--$5\times10^{-2}$. The loss of
   SD-DPC is comparable to that of NN-DPC, whereas the distilled policy on
   the same terms reaches about twice that of SD-DPC. What the sparse parameterization adds here is structure.
The policy found by SD-DPC (Fig.~\ref{fig:closed-loop-results}(b)) retains
only three terms,
\begin{equation}
\begin{aligned}
    v(x) &= \kappa_1 x_1 + \kappa_2 x_2 + \kappa_3 x_1^2 x_2,\\
    \kappa &= (-0.9873,\, -2.8860,\, 0.2716),
\end{aligned}
\label{eq:vdp-policy}
\end{equation}
with $u=\Pi_{[-5,5]}(v)$, and is therefore simple enough for {direct
analysis (Remark~\ref{rem:vdp}).}

{\begin{remark}\label{rem:vdp}
Because $F_\Xi$ and $\pi_{\Xi_\pi}$ are finite expansions over explicit
analytic libraries, the learned closed loop is available in closed form and
admits direct analysis in specific cases. For the policy
\eqref{eq:vdp-policy} applied to \eqref{eq:vdp} without saturation,
$V=\tfrac12(1-\kappa_1)x_1^2+\tfrac12x_2^2$ yields
$\dot V=(\mu+\kappa_2)x_2^2-(\mu-\kappa_3)x_1^2x_2^2\le0$ whenever
$\kappa_1<1$, $\kappa_2<-\mu$, and $\kappa_3\le\mu$, conditions satisfied by
the learned coefficients, so that the origin of this continuous-time closed loop is globally asymptotically stable \cite[Cor.~4.2]{Khalil2002}.
\end{remark}}

\subsection{Two-Dimensional Integrator with Obstacle Avoidance}
This case study shows why the policy must be trained in closed loop, with the
constraints enforced during optimization. SD-DPC retains twelve terms
(Fig.~\ref{fig:closed-loop-results}(c)), steers around the obstacle, and
reaches the setpoint from every sampled initial condition, with cumulative
control effort within $0.6\%$ of {the minimum attained by} IPOPT. SD-DPC and NN-DPC are both
trained through closed-loop rollouts that include the obstacle penalty. {Both satisfy the constraint, except in one SD-DPC scenario (depth $0.01$).} The terminal error of SD-DPC is higher than that of NN-DPC, the price of
a twelve-term parameterization, remains well within the reaching task. The distilled policy violates the
obstacle constraint in $60.8\%$ of the scenarios, with a maximum violation
depth of $1.61$, and its terminal error is about thirteen times that of SD-DPC.
Its coefficients are regressed onto the teacher's actions without the
obstacle penalty, {Proposition~\ref{prop:lipschitz}(iii), applied with $\pi$ the student
and $\tilde\pi$ the teacher on the same plant ($\delta=0$), explains why this
gives no safety guarantee. The imitation error
$\|\pi(x)-\tilde\pi(x)\|$ is at most $0.08$ on the states visited by the
teacher, where the regression data lie, but reaches $0.35$ on those visited
by the student, so $\varepsilon_\pi\ge0.35$. For Proposition~\ref{prop:lipschitz}(iii) to certify that the student
inherits the teacher's constraint satisfaction, the teacher would need a margin
$h_{\mathrm{obs}}\ge L_he_t\ge L_hL_u\varepsilon_\pi\ge1.05$, since $L_u=1$
for the integrator and $L_h\ge3$ outside the obstacle, where
$\|\nabla h_{\mathrm{obs}}\|=2\|x-x^o\|\ge3$. The near-optimal teacher, however,
clears the obstacle by only $h_{\mathrm{obs}}=0.06$, so no margin is left to absorb the imitation error.}

\section{Conclusions}
\label{sec:conclusions}
This work presented SD-DPC, a framework in which both the prediction model and
the feedback policy are sparse combinations of candidate functions, learned by
the same thresholded gradient-descent procedure, so that the policy terms
are chosen by the constrained closed-loop objective rather than by imitating a
pretrained controller. Across three case studies, the learned SD-DPC policies retained only a few terms,
satisfied the constraints in {almost} all test scenarios, and {reduced the task
loss of a policy distilled onto the same terms by up to an order of
magnitude, whereas the distilled policy violated the obstacle constraint.} The
closed-form policies require about two orders of magnitude less memory than
deep neural policies and {far less memory and evaluation time than the
optimization-based benchmark}. The sparse policies also
admit {explicit sensitivity bounds and, in specific cases, direct
stability analysis, as shown for the Van der Pol oscillator}. Constraint satisfaction is currently promoted through
penalty terms and verified empirically. Future work will focus on exploiting the symbolic structure of SD-DPC for stronger closed-loop guarantees. This includes hard feasibility through safety filters, systematic region-of-attraction certification, robustness bounds under identification error, and sufficient conditions under which sparsification preserves stability and feasibility.
We also plan to address noisy data and library mismatch, task-informed libraries, and online adaptation of the retained coefficients after deployment.

%%%%%%%%%%%%%%%%%%%%%%%%%%%%%%%%%%%%%%%%%%%%%%%%%%%%%%%%%%%%%%%%%%%%%%%%%%%%%%%%
% REFERENCES
%%%%%%%%%%%%%%%%%%%%%%%%%%%%%%%%%%%%%%%%%%%%%%%%%%%%%%%%%%%%%%%%%%%%%%%%%%%%%%%%

\begingroup
\makeatletter
\let\originalthebibliography\thebibliography
\renewcommand{\thebibliography}[1]{%
  \originalthebibliography{#1}%
  \setlength{\itemsep}{-0.6pt}%
  \interlinepenalty=0%
  \clubpenalty=0%
  \widowpenalty=0%
}
\makeatother

\bibliographystyle{IEEEtran}
\bibliography{IEEEabrv,autosam_ieee}

@IEEEtranBSTCTL{IEEEexample:BSTcontrol,
  CTLdash_repeated_names = "no"
}

@STRING{IEEE_J_SMCS  = "{IEEE} Trans. Syst., Man, Cybern. Syst."}

@STRING{IEEE_J_LCSS  = "{IEEE} Control Syst. Lett."}

@STRING{IEEE_J_CYB   = "{IEEE} Trans. Cybern."}

@article{Bemporad2002,
  author  = {Bemporad, Alberto and Morari, Manfred and Dua, Vivek and Pistikopoulos, Efstratios N.},
  title   = {The explicit linear quadratic regulator for constrained systems},
  journal = {Automatica},
  volume  = {38},
  number  = {1},
  pages   = {3--20},
  year    = {2002}
}

@article{Brunton2016a,
  author  = {Brunton, Steven L. and Proctor, Joshua L. and Kutz, J. Nathan},
  title   = {Discovering governing equations from data by sparse identification of nonlinear dynamical systems},
  journal = {Proc. Natl. Acad. Sci. USA},
  volume  = {113},
  number  = {15},
  pages   = {3932--3937},
  year    = {2016}
}

@article{brunton2016sparse,
  title={Sparse identification of nonlinear dynamics with control ({SINDYc})},
  author={Brunton, Steven L and Proctor, Joshua L and Kutz, J Nathan},
  journal={IFAC-PapersOnLine},
  volume={49},
  number={18},
  pages={710--715},
  year={2016},
  publisher={Elsevier}
}

@article{Drgona2022,
  author  = {Drgo{\v{n}}a, J{\'a}n and Ki{\v{s}}, Karol and Tuor, Aaron and Vrabie, Draguna and Klau{\v{c}}o, Martin},
  title   = {Differentiable predictive control: Deep learning alternative to explicit model predictive control for unknown nonlinear systems},
  journal = {J. Process Control},
  volume  = {116},
  pages   = {80--92},
  year    = {2022}
}

@article{Drgona2024,
  author  = {Drgo{\v{n}}a, J{\'a}n and Tuor, Aaron and Vrabie, Draguna},
  title   = {Learning constrained parametric differentiable predictive control policies with guarantees},
  journal = IEEE_J_SMCS,
  volume  = {54},
  number  = {6},
  pages   = {3596--3607},
  year    = {2024}
}

@article{Fasel2022,
  author  = {Fasel, Urban and Kutz, J. Nathan and Brunton, Bingni W. and Brunton, Steven L.},
  title   = {Ensemble-{SINDy}: Robust sparse model discovery in the low-data, high-noise limit, with active learning and control},
  journal = {Proc. R. Soc. A},
  volume  = {478},
  number  = {2260},
  pages   = {20210904},
  year    = {2022}
}

@article{GarciaPrett1989,
  title={Model predictive control: Theory and practice—A survey},
  author={Garcia, Carlos E and Prett, David M and Morari, Manfred},
  journal={Automatica},
  volume={25},
  number={3},
  pages={335--348},
  year={1989},
  publisher={Elsevier}
}

@article{fazlyab2019efficient,
  title={Efficient and accurate estimation of lipschitz constants for deep neural networks},
  author={Fazlyab, Mahyar and Robey, Alexander and Hassani, Hamed and Morari, Manfred and Pappas, George},
  journal={Advances in Neural Information Processing Systems},
  volume={32},
  year={2019}
}

@article{Hertneck2018,
  author  = {Hertneck, Michael and K{\"o}hler, Johannes and Trimpe, Sebastian and Allg{\"o}wer, Frank},
  title   = {Learning an approximate model predictive controller with guarantees},
  journal = IEEE_J_LCSS,
  volume  = {2},
  number  = {3},
  pages   = {543--548},
  year    = {2018}
}

@article{Kaiser2018,
  author  = {Kaiser, Eurika and Kutz, J. Nathan and Brunton, Steven L.},
  title   = {Sparse identification of nonlinear dynamics for model predictive control in the low-data limit},
  journal = {Proc. R. Soc. A},
  volume  = {474},
  number  = {2219},
  pages   = {20180335},
  year    = {2018}
}

@article{Karg2020,
  author  = {Karg, Benjamin and Lucia, Sergio},
  title   = {Efficient representation and approximation of model predictive control laws via deep learning},
  journal = IEEE_J_CYB,
  volume  = {50},
  number  = {9},
  pages   = {3866--3878},
  year    = {2020}
}

@inproceedings{Landajuela2021,
  title={Discovering symbolic policies with deep reinforcement learning},
  author={Landajuela, Mikel and Petersen, Brenden K and Kim, Sookyung and Santiago, Claudio P and Glatt, Ruben and Mundhenk, Nathan and Pettit, Jacob F and Faissol, Daniel},
  booktitle={International Conference on Machine Learning},
  pages={5979--5989},
  year={2021},
  organization={PMLR}
}

@article{Mayne2000,
  author  = {Mayne, David Q. and Rawlings, James B. and Rao, Christopher V. and Scokaert, Pierre O. M.},
  title   = {Constrained model predictive control: Stability and optimality},
  journal = {Automatica},
  volume  = {36},
  number  = {6},
  pages   = {789--814},
  year    = {2000}
}

@article{Messenger2021,
  author  = {Messenger, Daniel A. and Bortz, David M.},
  title   = {Weak {SINDy}: {Galerkin}-based data-driven model selection},
  journal = {Multiscale Model. Simul.},
  volume  = {19},
  number  = {3},
  pages   = {1474--1497},
  year    = {2021}
}

@article{Schwenzer2021,
  title={Review on model predictive control: An engineering perspective},
  author={Schwenzer, Max and Ay, Muzaffer and Bergs, Thomas and Abel, Dirk},
  journal={The international journal of advanced manufacturing technology},
  volume={117},
  number={5},
  pages={1327--1349},
  year={2021},
  publisher={Springer}
}

@article{LiHuaCao2022,
author  = {Li, Yun and Hua, Kaixun and Cao, Yankai},
title   = {Using stochastic programming to train neural network
approximation of nonlinear {MPC} laws},
journal = {Automatica},
volume  = {146},
pages   = {110665},
year    = {2022},
doi     = {10.1016/j.automatica.2022.110665}
}

@article{Kaheman2022,
  title={Automatic differentiation to simultaneously identify nonlinear dynamics and extract noise probability distributions from data},
  author={Kaheman, Kadierdan and Brunton, Steven L and Nathan Kutz, J},
  journal={Machine Learning: Science and Technology},
  volume={3},
  number={1},
  pages={015031},
  year={2022},
  publisher={IOP Publishing}
}

@article{ZhangTangHe2022,
  title={Learning chaotic systems from noisy data via multi-step optimization and adaptive training},
  author={Zhang, Lei and Tang, Shaoqiang and He, Guowei},
  journal={Chaos: An Interdisciplinary Journal of Nonlinear Science},
  volume={32},
  number={12},
  year={2022},
  publisher={AIP Publishing}
}

@article{Zolman2025,
  author  = {Zolman, Nicholas and Lagemann, Christian and Fasel, Urban and Kutz, J. Nathan and Brunton, Steven L.},
  title   = {{SINDy-RL} for interpretable and efficient model-based reinforcement learning},
  journal = {Nat. Commun.},
  volume  = {16},
  pages   = {10714},
  year    = {2025}
}

@article{Neuromancer2023,
  title={{NeuroMANCER: Neural Modules with Adaptive Nonlinear Constraints and Efficient Regularizations}},
  author={Drgona, Jan and Tuor, Aaron and Koch, James and Shapiro, Madelyn and Jacob, Bruno and Vrabie, Draguna},
  Url= {https://github.com/pnnl/neuromancer}, 
  year={2023}
}

@article{loshchilov2017decoupled,
  title={Decoupled weight decay regularization},
  author={Loshchilov, Ilya and Hutter, Frank},
  journal={arXiv preprint arXiv:1711.05101},
  year={2017}
}

@inproceedings{Guo2023,
  author    = {Guo, Jiaming and Zhang, Rui and Peng, Shaohui and Yi, Qi and Hu, Xing and Chen, Ruizhi and Du, Zidong and Zhang, Xishan and Li, Ling and Guo, Qi and Chen, Yunji},
  title     = {Efficient symbolic policy learning with differentiable symbolic expression},
  booktitle = {Advances in Neural Information Processing Systems},
  volume    = {36},
  pages     = {36278--36304},
  year      = {2023}
}

@inproceedings{Hose2024,
  author    = {Hose, Henrik and Gr{\"a}fe, Alexander and Trimpe, Sebastian},
  title     = {Parameter-adaptive approximate {MPC}: Tuning neural-network controllers without retraining},
  booktitle = {Proc. 6th Annual Learning for Dynamics \& Control Conference},
  series    = {Proceedings of Machine Learning Research},
  volume    = {242},
  pages     = {349--360},
  year      = {2024}
}

@article{wachter2006implementation,
  author  = {W{\"a}chter, Andreas and Biegler, Lorenz T.},
  title   = {On the implementation of an interior-point filter line-search algorithm for large-scale nonlinear programming},
  journal = {Math. Program.},
  volume  = {106},
  number  = {1},
  pages   = {25--57},
  year    = {2006}
}

@Article{Andersson2019,
  author = {Joel A E Andersson and Joris Gillis and Greg Horn
            and James B Rawlings and Moritz Diehl},
  title = {{CasADi} -- {A} software framework for nonlinear optimization
           and optimal control},
  journal = {Mathematical Programming Computation},
  volume = {11},
  number = {1},
  pages = {1--36},
  year = {2019},
  publisher = {Springer},
  doi = {10.1007/s12532-018-0139-4}
}

@article{Champion2019,
  author  = {Champion, Kathleen and Lusch, Bethany and Kutz, J. Nathan and Brunton, Steven L.},
  title   = {Data-driven discovery of coordinates and governing equations},
  journal = {Proc. Natl. Acad. Sci. USA},
  volume  = {116},
  number  = {45},
  pages   = {22445--22451},
  year    = {2019},
  doi     = {10.1073/pnas.1906995116}
}

@inproceedings{Botteghi2024,
  author    = {Botteghi, Nicol{\`o} and Fasel, Urban},
  title     = {Parametric {PDE} control with deep reinforcement learning and $\ell_0$ sparse polynomial policies},
  booktitle = {Proc. {IEEE} 63rd Conference on Decision and Control ({CDC})},
  pages     = {6108--6115},
  year      = {2024},
  doi       = {10.1109/CDC56724.2024.10886215}
}

@article{Kunisch2023,
  title={Learning optimal feedback operators and their sparse polynomial approximations},
  author={Kunisch, Karl and V{\'a}squez-Varas, Donato and Walter, Daniel},
  journal={Journal of Machine Learning Research},
  volume={24},
  number={301},
  pages={1--38},
  year={2023}
}

@misc{Bengio2013,
  author       = {Bengio, Yoshua and L{\'e}onard, Nicholas and Courville, Aaron},
  title        = {Estimating or propagating gradients through stochastic neurons for conditional computation},
  howpublished = {arXiv preprint arXiv:1308.3432},
  year         = {2013},
  url          = {https://arxiv.org/abs/1308.3432}
}

@book{Khalil2002,
  author    = {Khalil, Hassan K.},
  title     = {Nonlinear Systems},
  edition   = {Third},
  publisher = {Prentice Hall},
  address   = {Upper Saddle River, NJ},
  year      = {2002}
}
\endgroup
\end{document}